\documentclass[pra,aps,twocolumn,superscriptaddress,dvipsnames,amssymb,9pt,floatfix]{revtex4}
	
\usepackage{graphicx}
\usepackage{amsmath,amssymb,amsthm,amsfonts,mathptmx}
\usepackage{color}
\usepackage[matrix,frame,arrow]{xypic}
\usepackage{times}
\usepackage{algorithm}
\usepackage{algcompatible}
\usepackage{longtable}
\usepackage{tikz}

\theoremstyle{plain}

\newtheorem{theorem}{Theorem}

\newtheorem{lemma}{Lemma}
\theoremstyle{definition}

\newcommand{\bra}[1]{\langle#1|}
\newcommand{\ket}[1]{|#1\rangle}

\DeclareMathOperator{\Tr}{Tr}

\begin{document}

\title{A measurement driven analog of adiabatic quantum computation for frustration-free Hamiltonians}
\author{Liming Zhao}
\affiliation{Singapore University of Technology and Design,  8 Somapah Road, Singapore 487372}
\affiliation{Centre for Quantum Technologies, National University of Singapore, 3 Science Drive 2, Singapore 117543}
\author{Carlos A. P\'erez-Delgado}
\affiliation{School of Computing, University of Kent, Canterbury CT2 7NF, United Kingdom}
\author{Simon C. Benjamin}
\affiliation{Department of Materials, University of Oxford, Parks Rd, Oxford OX1 3PH, United Kingdom}
\author{Joseph F. Fitzsimons}\email{joseph\_fitzsimons@sutd.edu.sg}
\affiliation{Singapore University of Technology and Design,  8 Somapah Road, Singapore 487372}
\affiliation{Centre for Quantum Technologies, National University of Singapore, 3 Science Drive 2, Singapore 117543}
\begin{abstract}
The adiabatic quantum algorithm has drawn intense interest as a potential approach to accelerating optimization tasks using quantum computation. The algorithm is most naturally realised in systems which support Hamiltonian evolution, rather than discrete gates. We explore an alternative approach in which slowly varying measurements are used to mimic adiabatic evolution. We show that for certain Hamiltonians, which remain frustration-free all along the adiabatic path, the necessary measurements can be implemented through the measurement of random terms from the Hamiltonian. This offers a new, and potentially more viable, method of realising adiabatic evolution in gate-based quantum computer architectures. 
\end{abstract}
\maketitle


In the field of quantum computation, it has long been recognized that there exists deep connections between ground states of Hamiltonians and  problems of fundamental interest to the study of computational complexity \cite{barahona1982computational,wille1985computational}. It is known that the problem of finding the ground state of a Hamiltonian is hard even in the case of one-dimensional lattices \cite{aharonov2009power}, and that in general the $k$-local Hamiltonian problem is QMA-hard (and hence NP-hard) for any $k\geq 2$ \cite{kitaev1999quantum,kempe2006complexity}. Over the years many classical and, more recently, quantum algorithms have been proposed to address this problem \cite{white1993density,shao2006advances,verstraete2008matrix,schollwock2011density,landau2015polynomial,abrams1999quantum,aspuru2005simulated}. While all polynomial time algorithms are destined to fail, under the assumption that P$\neq$NP, such algorithms often work for Hamiltonians of practical interest.

One such quantum algorithm is the adiabatic algorithm \cite{farhi2000quantum}, which is fundamentally rooted in the adiabatic theorem \cite{born1928beweis}. Informally, the adiabatic theorem states that a system starting in the ground state of some initial Hamiltonian will stay close to the ground state of the system if the Hamiltonian is gradually changed over time, provided that this change is continuous and sufficiently slow. This means that one can prepare the ground state of an arbitrary Hamiltonian $\mathcal{H}_f$ by first preparing the ground state of some simple Hamiltonian $\mathcal{H}_I$ and then subjecting the system to a time varying Hamiltonian which slowly interpolates between $\mathcal{H}_I$ and $ \mathcal{H}_F$. In its simplest form, the adiabatic algorithm considers a linear interpolation between the initial and final Hamiltonians described by $\mathcal{H}(s)=(1-s)\mathcal{H}_I+s  \mathcal{H}_F$ for $ s \in [0,1]$, where $s$ is some simple function of time. This provides a heuristic approach for tackling satisfiability problems \cite{farhi2000quantum,farhi2001quantum}. In general, the timescale required for this evolution can be exponentially long, as it scales with the reciprocal of the gap between the ground state and first excited state of the instantaneous Hamiltonians at each point in time. This reconciles the adiabatic approach with the fact that QMA is not known to be contained in BQP, the class of problems efficiently solvable on a quantum computer. Indeed, it is now known that the adiabatic model is equivalent to circuit model quantum computation \cite{aharonov2008adiabatic}.

Due to its wide applicability as a black-box optimization technique, the adiabatic algorithm and similar techniques such as quantum annealing have emerged as one of the key use-cases for quantum processors \cite{santoro2006optimization}. The efficient implementation of such techniques raises architectural concerns, however. While adiabatic evolution is in principle possible in many monolithic quantum processor architectures, the Hamiltonians possible are often restricted to 2-local interactions according to some fixed graph \cite{bunyk2014architectural}. While techniques have been devised to overcome these limitations, they incur significant overhead \cite{lechner2015quantum,rocchetto2016stabilizers,chancellor2016direct}. The situation is far worse when one considers the case of distributed quantum computing architectures, such as many promising ion-trap and quantum dot proposals \cite{kielpinski2002architecture,loss1998quantum}, which implement entangling operations between nodes using discrete operations rather than Hamiltonian dynamics. For such systems, a direct implementation of adiabatic computation requires simulating Hamiltonian dynamics with discrete logic gates, an approach which would incur prohibitive overhead \cite{wecker2014gate}.
 
Here we show that it is possible to implement adiabatic-like evolution using relatively simple measurements provided that the Hamiltonian remains frustration free at all points along the adiabtic path. Our results are based on a connection between the adiabatic theorem and the quantum Zeno effect \cite{burgarth2013non}. We begin by presenting an alternate proof of a result due to Somma, Boixo and Knill which gave an adiabatic-like theorem for systems measured (or dephased) in the eigenbases of slowly varying Hamiltonians. We then show that for frustration-free Hamiltonians, measurement of randomly chosen individual terms of the Hamiltonian suffices to approximate measurement of the ground state, satisfying our criterion for adiabatic-like evolution. For $k$-local Hamiltonians, these measurement have constant complexity, as they correspond to projectors on at most $k$ qubits. This potentially opens the door to a direct analogue of the adiabatic algorithm well suited for distributed architectures, such as ion-trap implementations and similar systems currently under investigation \cite{monroe2013scaling,nickerson2014freely}. These results also provide some level of theoretical understanding of the mechanism behind a measurement-driven approach to SAT-solving proposed by Benjamin \cite{benjamin2015performance} which has shown promising performance in numerical experiments.

We start by considering the evolution of the state of a quantum system due to the measurement of a sequence of observables, which we treat as corresponding directly to Hamiltonians. We then prove that, provided the difference between pairs of neighbouring Hamiltonians in the sequence has sufficiently small norm compared to the energy gap between the ground state and first excited state, a system prepared in the ground state of the initial Hamiltonian will evolve to the ground state of the final Hamiltonian with high probability.

Let $\mathcal{H}_I$ and $ \mathcal{H}_F$ be the initial and final Hamiltonian respectively. Also, let $\{\mathcal{H}_n\}_{0 \leq n \leq N}$ be an ordered set of intermediate interpolating operators, such that $\mathcal{H}_0 \equiv \mathcal{H}_I$ and $\mathcal{H}_N \equiv  \mathcal{H}_F$. For simplicity, we will assume that every $\mathcal{H}_n$ is normalized such that the eigenvalues lie in the range between $0$ and $1$, with the lowest eigenvalue being exactly $0$. The assumption on the range of the eigenvalues can be made without loss of generality, as the Hamiltonians can always be rescaled by multiplying by a constant and shifted by adding a multiple of the identity. We will make no assumption regarding the degeneracy of the ground state space. Taking $\ket{\psi_0}$ to be a state in the ground state space of $\mathcal{H}_I$, and taking $\ket{\psi_{n}}$ to denote the normalized projection of $\ket{\psi_{n-1}}$ onto the ground state space of $\mathcal{H}_n$, the evolution of the system then satisfies the following constraint.

\begin{theorem}\label{thm:1}
Given a system initially in state $\ket{\psi_0}$, the state $\ket{\psi_N}$ can be obtained with probability $p \geq 1 - \epsilon$ by measuring the operators  $\mathcal{H}_n$ in sequence for $1 \leq n \leq N$, provided that
 \begin{equation}\label{eq:thm1-cond}
 \max_{1 \leq n \leq N} \left( \frac{\parallel \Delta \mathcal{H}_n \parallel^2_\infty}{{g(\mathcal{H}_n)}^2 } \right) \leq \frac{\epsilon}{N},\nonumber
\end{equation}
where $g(\mathcal{H}_n)$ is the gap between the eigenvalues corresponding to the ground state space and first excited state of $\mathcal{H}_n$, and $\Delta \mathcal{H}_n=\mathcal{H}_n-\mathcal{H}_{n-1}$. Furthermore, if at each step $n$ the measurement of $\mathcal{H}_n$ is replaced with any procedure that produces a state $\rho_{n}$, such that the trace distance from $\ket{\psi_n}\bra{\psi_n}$ is at most $\frac{\delta}{2N}$, with probability at least $\bra{\psi_{n}}\rho_{n-1}\ket{\psi_n}$, then the overall procedure yields a state $\rho_N$, with trace distance at most $\frac{\delta}{2N}$ from $\ket{\psi_N}\bra{\psi_N}$, with probability $p'\geq 1 - \epsilon -\delta$.
\end{theorem}
\begin{proof}
Taking $P_n$ to be the projector onto the ground state space of $\mathcal{H}_n$, then the probability of successfully obtaining $\ket{\psi_n}$ from $\ket{\psi_{n-1}}$ is given by $p_n = \|P_n \ket{\psi_{n-1}}\|^2$. Then, the probability of successfully projecting onto $\ket{\psi_N}$ during the final measurement is bounded by $p \geq \textstyle\prod_n p_n$. The reason this is a bound rather than an exact equality is due to the possibility of reaching the correct final state through a sequence of measurements fails to project onto the ground state of some intermediate Hamiltonian.

Now, consider the probability of failure at step $n$, assuming that all previous measurements have successfully projected onto the ground state space of the associated Hamiltonian,
\begin{equation}
\epsilon_n = \| \left(I-P_n \right) \vert\psi_{n-1}\rangle \|^2.\nonumber
\end{equation}
This can be turned into an inequality by making use of Loewner order, noting that $\left(I-P_n\right) \leq\frac{\mathcal{H}_n}{g(\mathcal{H}_n)}I$, and hence
\begin{equation}
\epsilon_n \leq \left\| \frac{\mathcal{H}_n}{g(\mathcal{H}_n)} \vert\psi_{n-1}\rangle \right\|^2 = \left\| \frac{\Delta \mathcal{H}_n}{g(\mathcal{H}_n)} \vert\psi_{n-1}\rangle \right\|^2,\nonumber
\end{equation}
where the equality follows from the fact that $\mathcal{H}_{n-1} \ket{\psi_{n-1}} = 0$.
This can be used to bound $p_n$. By making use of the definition of the infinity norm for matrices, we arrive at
 \begin{equation}
p_n \geq 1- \frac{\parallel \Delta \mathcal{H}_n \parallel^2_\infty}{{g(\mathcal{H}_n)}^2 }.  \nonumber
\end{equation}
The final success probability is then bounded by
 \begin{equation}
p \geq  1- \sum_{n=1}^{N} \frac{\parallel \Delta \mathcal{H}_n \parallel^2_\infty}{{  g(\mathcal{H}_n)}^2 }.\nonumber
\end{equation}
Provided that Eq.~\ref{eq:thm1-cond} holds, we then have $p \geq 1 - \epsilon$ as required.

When considering the modified procedure, the modified probability of success at each step is bounded from below by $p'_n \geq \Tr(P_n \rho_{n-1})$. This can be rewritten as  $p'_n = p_n + \Tr(P_n(\rho_{n-1}-\ket{\psi_{n-1}}\bra{\psi_{n-1}}))$. Using the trace distance constraint, this implies $p'_n \geq 1-\frac{\epsilon}{N} - \frac{\delta}{N}$ and hence $p' \geq 1-\epsilon - \delta$ as required.
\end{proof} 

While Eq.~\ref{eq:thm1-cond} may appear unusual when compared to adiabatic conditions, due to the way in which $N$ appears as a reciprocal it can be transformed into a more conventional form by making the substitution $\delta_N \mathcal{H}_n = N \Delta \mathcal{H}_n$, to obtain
\begin{eqnarray}
N \geq \epsilon^{-1} \max_{1 \leq n \leq N} \left( \frac{\parallel \delta_N \mathcal{H}_n \parallel^2_\infty}{{g(\mathcal{H}_n)}^2 } \right).\nonumber
\end{eqnarray}

Suppose that for any $N$ each of the measured Hamiltonians $\mathcal{H}_n$ is chosen along a fixed continuous path $\mathcal{H}(s)$, for $0\leq s \leq 1$, through the space of Hamiltonians, such that they lie sequentially along this path at equal intervals. In this case, for large $N$ the finite difference $\delta_N \mathcal{H}_n$ tends to the derivative $\frac{d}{ds} \mathcal{H}(s)$, and is thus approximately constant for large $N$, depending only on the path through the space of Hamiltonians. Note that $N$ does not have dimensions of time, and so this equation is not directly comparable to adiabatic theorems. However, making the substitution $T = N/\max_{1 \leq n \leq N} \delta_n \mathcal{H}_n$ one obtains a more conventional adiabatic expression (similar to that in Ref.~\cite{van2001powerful}).

While the result presented above provides a link between the measurement of interpolating Hamiltonians and the adiabatic theorem, this does not imply that measurements are a viable alternative to Hamiltonian evolution for implementing adiabatic quantum computation. After all, the measurement of a Hamiltonian is a non-trivial task, and implementing it via controlled unitary evolution and phase estimation \cite{kitaev1995quantum} may provide little advantage over directly implementing adiabatic evolution. In order to increase the utility of this correspondence, we now introduce a method for efficiently projecting onto the ground state of frustration-free Hamiltonians. 

Let $ \mathcal{H}$ be a frustration-free Hamiltonian which is the sum of $m$ terms,
\begin{equation}
 \mathcal{H}=\textstyle\sum_{i=1}^{m} \omega_i H_{i}, \nonumber
\end{equation}
where every term $H_{i}$ is a tensor product of $2\times 2$ Hermitian operators. We assume that $\sum_{i=1}^{m} \omega_i =1 $ and the eigenvalues of each term is between $0$ and $1$, which can be done without loss of generality as discussed earlier. For each $H_{i}$ one can construct a POVM measurement with measurement operators $ E_i = \sqrt{I - H_{i}}$ and $\tilde E_i = \sqrt{H_{i}}$. Specifically, if the eigenvalues are either $0$ or $1$, one can construct a projective measurement with projectors $H_i$ and $I-H_i$. The lowest energy subspace is obtained when the measurement result is ${I - H_{i}}$.

Now, consider the following operation $M$ on an arbitrary quantum state $\rho$. First, an index $1 \leq i \leq m$ is selected at random with probability $\omega_i$. A POVM measurement is then performed on $\rho$ with measurement operators $ E_i$ and $\tilde E_i$. If the outcome of the measurement corresponds to application of $\tilde E_i$ then the procedure is said to fail. Otherwise, the resulting state of the system is $\rho'_i = \frac{E_i \rho  E_i^\dagger}{\text{Tr}(E_i \rho  E_i^\dagger)}$. This latter case occurs with probability $p(s|i) = \text{Tr}(E_i \rho  E_i^\dagger)$. Disregarding the choice of $i$, the output state $\rho'$ of a successful application of $M$ will be a mixed state consisting of a distribution over the various possibilities for $\rho'_i$ as follows. Let $p(s)$ be the total success probability. Since every $i$ is chosen with probability $p(i) = \omega_i$, we then have
\begin{align}
p(s) &=   \sum_{i=1}^m  \omega_i  \Tr\left( E_i \rho E_i ^\dagger\right) \nonumber\\
&= \Tr \left(\sum_{i=1}^m  \omega_i  \left(I - H_{i}\right) \rho\right)\nonumber\\
&=  1 - \Tr \left(\mathcal{H}\rho\right).\label{eq:success_prob}
\end{align}
From Bayes' theorem, the output state $\rho'$ is then given by
\begin{equation}
\rho' =   \sum_{i=1}^{m} \frac{p(i) p(s | i)}{p(s)} \rho'_i  = \frac{1}{1 - \Tr \left(\mathcal{H}\rho\right)} \sum_{i=1}^{m}\omega_i  E_i \rho E_i^\dag . \nonumber
 \end{equation}

We now show that successful application of the operation $M$ to a state $\rho$, with non-zero overlap with the ground state space, will increase the projection onto the ground state space.
 
\begin{lemma}\label{lemma:1}
	Let $\mathcal{H}$ be a frustration-free Hamiltonian, as described above. Let $P_\text{gs}$ be the projector onto the ground state space of $\mathcal{H}$. Let $\rho$ be an arbitrary density matrix and let $\rho'$ be the resulting density matrix after a successful application of the operation $M$ as defined above to $\rho$. Then,
\begin{equation}
\Tr\left(P_\text{gs} \rho'\right) = \frac{\Tr\left(P_\text{gs} \rho\right)}{1 - \Tr\left(\mathcal{H} \rho \right)},\label{eq:thm1res}
\end{equation}
and the probability that $M$ is successful is $1 - \Tr \left(\mathcal{H} \rho\right)$.
\end{lemma}
\begin{proof}
We begin by noting that 
\begin{align}
\Tr\left(P_\text{gs} \rho'\right) &=\frac{1}{1 - \Tr \left(\mathcal{H}\rho\right)}  \Tr\left( P_\text{gs} \sum_{i=1}^{m} \omega_i  E_i \rho E_i^\dag \right).\nonumber
\end{align}
Using the cyclic property of trace, this can be rewritten as 
\begin{equation}
\Tr\left(P_\text{gs} \rho'\right) =\frac{1}{1 - \Tr \left(\mathcal{H}\rho\right)}   \sum_{i=1}^{m} \omega_i \Tr\left(E_i P_\text{gs}  E_i^\dag \rho \right).
\end{equation}
The measurement operators can then be absorbed into $P_\text{gs}$. Evaluating the summation then yields Eq.~\ref{eq:thm1res} as required. The probability of success for applying $M$ was previously calculated in Eq.~\ref{eq:success_prob}.
\end{proof}

We now consider what happens when $M$ is applied not once, but some number of times $k$.
 \begin{theorem}\label{thm:2}
Let $\mathcal{H}$ be a frustration-free Hamiltonian. Let $P_\text{gs}$ be the projector onto the ground state space of $\mathcal{H}$. Let $\rho$ be a density matrix with non-zero overlap with the ground state space of $\mathcal{H}$ and let $\rho^{(k)}$ be the resulting density matrix after a successful application of the operation $M$ as defined above to $\rho$ sequentially $k$ times. Then,
  \begin{equation*}
 	\Tr\left(P_\text{gs}\rho^{(k)}\right) \geq \left(1+\left(1-g(\mathcal{H})\right)^k \left(\frac{1}{\Tr\left(P_\text{gs}\rho \right)} - 1\right)\right)^{-1}
 \end{equation*}
Furthermore, $P_\text{gs} \rho^{(k)} P_\text{gs} \propto P_\text{gs}\rho P_\text{gs}$ and the probability that all $k$ applications of $M$ are successful is at least $\Tr\left(P_\text{gs}\rho \right)$.
\end{theorem}
 \begin{proof} 
 We will consider the ratio 
 \begin{equation}
 	R_\ell = \frac{\Tr\left(\left(I-P_\text{gs}\right)\rho^{(\ell)}\right)}{\Tr\left(P_\text{gs}\rho^{(\ell)}\right)} = \Tr\left(P_\text{gs}\rho^{(\ell)}\right)^{-1} - 1.\label{eq:rldef}
 \end{equation}
 By definition $\rho^{(\ell)} = M\left(\rho^{(\ell-1)}\right)$ for all $\ell > 1$, and hence from Lemma~\ref{lemma:1} it follows that
 \begin{align}
 	R_\ell &= \left(1-\Tr\left(\mathcal{H}\rho^{(\ell-1)}\right)\right)\Tr\left(P_\text{gs}\rho^{(\ell-1)}\right)^{-1} - 1 .\nonumber
 \end{align}
Since $\Tr\left(\mathcal{H}\rho^{(\ell-1)}\right) \geq g(\mathcal{H})\left(1-\Tr\left(P_\text{gs}\rho^{(\ell-1)}\right)\right)$ this gives rise to the bound
\begin{align}
 	R_\ell &\leq \left(1-g(\mathcal{H})\right)\left(\Tr\left(P_\text{gs}\rho^{(\ell-1)}\right)^{-1} - 1\right).\nonumber
 \end{align}
Using Eq.~\ref{eq:rldef} we then arrive at the recurrence inequality 
 \begin{equation}
 	R_\ell \leq \left(1-g(\mathcal{H})\right) R_{\ell-1}.\nonumber
 \end{equation}
 Hence $R_k \leq \left(1-g(\mathcal{H})\right)^k R_0$. From Eq.~\ref{eq:rldef} we can then replace $R_k$ and $R_0$ to obtain
  \begin{equation*}
 	\Tr\left(P_\text{gs}\rho^{(k)}\right) \geq \left(1+\left(1-g(\mathcal{H})\right)^k \left(\frac{1}{\Tr\left(P_\text{gs}\rho \right)} - 1\right)\right)^{-1}\nonumber
 \end{equation*}
 as required.
 
Turning to the projection of $\rho^{(k)}$ onto the ground state space, from the definition of $M$ we have 
 \begin{align}
	\rho^{(k)} &= \frac{\sum_{i_1\ldots i_k} E_{i_k} \ldots E_{i_1} \rho E_{i_1}^\dagger \ldots E_{i_k}^\dagger}{\sum_{j_1\ldots j_k} \Tr\left(E_{j_k} \ldots E_{j_1} \rho E_{j_1}^\dagger \ldots E_{j_k}^\dagger\right)}\nonumber
\end{align}
and hence
\begin{align}
	P_\text{gs}\rho^{(k)}P_\text{gs} &\propto \sum_{i_1\ldots i_k} P_\text{gs} E_{i_k} \ldots E_{i_1} \rho E_{i_1}^\dagger \ldots E_{i_k}^\dagger P_\text{gs} \nonumber\\
	&= \sum_{i_1\ldots i_k}  E_{i_k} \ldots E_{i_1} P_\text{gs}\rho P_\text{gs} E_{i_1}^\dagger \ldots E_{i_k}^\dagger \nonumber\\
		&= P_\text{gs}\rho P_\text{gs}.\nonumber
\end{align}

The success probability for applying $M$ any number of times can be lower bounded by noting that $M$ does not alter states in the ground state space of $\mathcal{H}$. Hence the trace of the projection of $\rho$ onto this subspace provides a lower bound.
  \end{proof}

Theorem \ref{thm:2} implies that applying $M$ sufficiently many times satisfies the requirements of Theorem \ref{thm:1} for a procedure approximately projecting onto the ground state space of a Hamiltonian. This can be made quantitative by noting that if $k = \alpha/g(\mathcal{H})$ then $\left(1-g(\mathcal{H})\right)^k  \leq e^{-\alpha}$. When used in the context of Theorem \ref{thm:1} it will necessarily be the case that $\left(\frac{1}{\Tr\left(P_\text{gs}\rho \right)} - 1\right) \ll 1$. In such cases it should suffice to choose $\alpha \propto \log N$ to provide the necessary accuracy. 

The results presented above hold even for Hamiltonians with degenerate ground states and thus are broadly applicable. The combination of these results provides a means for implementing adiabatic-like dynamics using measurements of only modest complexity, at least for frustration-free Hamiltonians. This suggests that such evolution can be realised without need for Trotterisation of Hamiltonian dynamics, and provides a potentially more viable approach in quantum computers based on discrete gates, especially in the context of distributed architectures. The restriction to frustration free Hamiltonians is used to ensure that the ground state is simultaneously an eigenstate of each possible measurement. Removing this restriction represents an interesting avenue for future research.

\acknowledgments

The authors thank Daniel Burgarth and Yingkai Ouyang for helpful discussions. This material is based on research supported in part by the Singapore National Research Foundation under NRF Award No.~NRF-NRFF2013-01. JFF acknowledges support from the Air Force Office of Scientific Research under AOARD grant no.~FA2386-15-1-4082. SCB acknowledges support from EPSRC grant EP/M013243/1.

 \bibliographystyle{apsrev}
 \bibliography{mybib}

\end{document}